\documentclass[journel]{IEEEtran}

\usepackage{latexsym}

\usepackage{amsmath,amssymb,amsfonts,verbatim,cite}
\def\F {{\mathbb{F}}}

\def\GRS{{\mathcal GRS}}
\def\ba{{\bf a}}
\def\bb{{\bf b}}
\def\bc{{\bf c}}

\def\bbc{\bar{\bc}}
\def\bu{{\bf u}}
\def\bv{{\bf v}}
\def\bw{{\bf w}}
\def\bo{{\bf 0}}
\def\bi{{\bf 1}}

\def\bx{{\bf x}}

\def\a{{\alpha}}
\def\Ga{{\alpha}}
\def\Gb{{\beta}}

\def\res{{\rm res}}

\def\beq{\begin{equation}}
\def\eeq{\end{equation}}

\newtheorem{thm}{Theorem}[section]

\newtheorem{lem}[thm]{Lemma}
\newtheorem{cor}[thm]{Corollary}
\numberwithin{equation}{section} 
\newtheorem{exm}[thm]{Example}
\newtheorem{rmk}[thm]{Remark}
\begin{document}

\onecolumn

\title{New  MDS Self-Dual Codes from Generalized Reed-Solomon Codes }

\author{Lingfei Jin and Chaoping Xing

\thanks{Lingfei Jin is with Shanghai Key Laboratory of Intelligent Information Processing, School of Computer Science, Fudan University, Shanghai 200433, China. {\it email:} {lfjin@fudan.edu.cn}. }
\thanks{Chaoping Xing is with Division of
Mathematical Sciences, School of Physical \& Mathematical Sciences,
Nanyang Technological University, Singapore 637371, Republic of
Singapore. {\it email:} {xingcp@ntu.edu.sg}
}
\thanks{Lingfei Jin is supported by the Open Research Fund of National Mobile Communications Research Laboratory, Southeast University (No. 2017D07), and by the National Natural Science Foundation of China under Grant 11501117.}
\thanks{C. Xing is supported by the Singapore Ministry of Education Tier 1 under Grant
RG20/13.
}}

\maketitle

\begin{abstract} Both MDS and Euclidean self-dual codes have theoretical and practical importance and the study of MDS self-dual codes has attracted lots of attention in recent years.
In particular, determining existence of $q$-ary MDS self-dual codes for various lengths has been investigated extensively. The problem is  completely solved for the case where $q$ is even.  The current paper focuses on the case where $q$ is odd. We construct a few classes of new MDS self-dual codes through generalized Reed-Solomon codes. More precisely, we show that for any given even length $n$ we have a $q$-ary MDS code as long as $q\equiv1\bmod{4}$ and $q$ is sufficiently large (say $q\ge 4^n\times n^2)$. Furthermore, we prove that there exists a $q$-ary MDS self-dual code of length $n$ if $q=r^2$ and $n$ satisfies one of the three conditions: (i) $n\le r$ and $n$ is even; (ii) $q$ is odd and $n-1$ is an odd divisor of $q-1$;   (iii) $r\equiv3\mod{4}$ and $n=2tr$ for any $t\le (r-1)/2$.
\end{abstract}

\begin{keywords}
Self-dual codes, MDS codes, Generalized Reed-Solomon codes.
\end{keywords}

\section{Introduction}
MDS codes and Euclidean self-dual codes belong to two different categories of block codes. Both classes are of practical and theoretical importance. In recent years,
study of MDS self-dual codes (we only consider Euclidean inner product in the following context) has attracted a lot of attention \cite{AKL08,BBDW04,BGGHK03,GK02,GG08,G12,HK06,KL04a,KL04}. First of all, MDS codes achieve optimal parameters that allow correction of maximal number of errors for a given code rate. Study of various properties of MDS codes, such as classification \cite{KKP15,PD91} of MDS codes, non-Reed-Solomon MDS codes \cite{RL89}, balanced MDS codes \cite{DSDY}, lowest density MDS codes \cite{BR99,LR06} and existence of MDS codes \cite{DSY14}, has been the center of the area.   In addition, MDS codes are closely connected to combinatorial design and finite geometry \cite[Chapters 11 and 14]{LX}. Furthermore, the generalized Reed-Solomon codes are a class of MDS codes and have found wide applications in practice. On the other hand, due to their nice structures, self-dual codes have been  attracting attention from both coding theorists, cryptographers and mathematicians. Self-dual codes have found various applications in cryptography (in particular secret sharing) \cite{CD08,DMP08,M95}  and combinatorics \cite{LX}.  Thus, it is natural to consider the intersection of these two classes, namely, MDS self-dual codes.

As the parameters of an MDS self-dual code is completely determined by its length,
one of the central problems in this topic is
to determine existence of MDS self-dual codes for various lengths. The problem is  completely solved for the case where $q$ is even \cite{GG08}.  The current paper focuses on the case where $q$ is odd. Our idea is to construct generalized Reed-Solomon code that are self-dual. Thus, the result is of theoretical interest and practical relevance.
\subsection{Known results}
One of the existing constructions of  MDS self-dual codes in literature is through constacyclic codes \cite{AKL08,GG08,KL04} because the generator polynomial of the dual code of a constacyclic code can be determined by the generator polynomial of the code. Some other approaches include orthogonal designs \cite{BGGHK03,GK02} and generalized Reed-Solomon codes  \cite{AKL08}. We summarize some  known results in the  Table I.


\begin{table}
\caption{Known results on existence of $q$-ary MDS self-dual codes of even length $n$ }
\label{table:1}
\begin{center}
\begin{tabular}{||c|c|c|c|c|c||}\hline
$q$ & $2|q$ & $2\not|q$  & $q=r^t$, $t$ even &$q=r^t$, $r\equiv3\bmod{4}$, $t$ odd &$q=r^t$, $r\equiv1\bmod{4}$, $t$ odd\\  \hline
$n$ & $n\le q$ & $n=q+1$ & $(n-1)|(r-1)$ & $n=p^m+1$, odd $m$ and prime $p$  & $n=p^m+1$, $m$ odd and prime $p$  \\
&&& & and $p\equiv3\bmod{4}$ &and $p\equiv1\bmod{4}$ \\  \hline
Reference &\cite{GG08}&\cite{GG08}&\cite{G12}&\cite{G12}&\cite{G12}\\ \hline
\end{tabular}
\end{center}
\end{table}

Besides the  results in Table I, only some sparse  lengths $n$ of MDS self-dual codes  have been found (see \cite{AKL08,BBDW04,BGGHK03,GK02,HK06,KL04}).

\subsection{Our results}
We show the following result in this paper.

\begin{thm}[Main Theorem]\label{thm:1.1} Let $q$ be an odd prime power and let $n$ be an even positive integer. Then there exists a $q$-ary MDS self-dual code of length $n$ if $q$ and $n$ satisfy one of the following conditions
\begin{itemize}
\item[{\rm (i)}] $q\equiv1\bmod{4}$ and  $q\ge 4^n\times n^2$ {\rm (see Theorem \ref{thm:3.2}(ii))};
  \item[{\rm (ii)}] $q=r^2$ and $n\le r$ {\rm (see Theorem \ref{thm:3.4}(i))};
 \item[{\rm (iii)}]  $q=r^2$ and $n-1$ is a divisor of $q-1$ {\rm (see Theorem \ref{thm:3.4}(ii))};
   \item[{\rm (iv)}] $q=r^2$, $r\equiv3\mod{4}$ and $n=2tr$ for any $t\le (r-1)/2$ {\rm (see Theorem \ref{thm:3.5})}.
\end{itemize}
\end{thm}
\begin{rmk}{\rm
Part (i) of Theorem \ref{thm:1.1} says that for any given even length $n$ we have a $q$-ary MDS code as long as $q\equiv1\bmod{4}$ and $q$ is sufficiently large (say $q\ge 4^n\times n^2)$, while Part (iii)  of Theorem \ref{thm:1.1} extends the result of \cite{G12} where a stricter  condition $(n-1)|(r-1)$ is required. In addition, we also use our approach to get MDS self-dual codes in \cite{GG08}. Note that the approach in  \cite{GG08} is quite different as the main tool for constructing MDS codes in \cite{GG08} is orthogonal design, while our construction is through generalized Reed-Solomon codes. }\end{rmk}

\subsection*{Our techniques} Our idea of constructing MDS self-dual codes is through  generalized Reed-Solomon (GRS or generalized RS for short) codes. In this paper, we  present two methods to construct generalized RS codes that are  self-dual. The first one is to directly find  elements $\Ga_1,\Ga_2,\dots,\Ga_n$ such that $\prod_{1\le j\le n, j\neq i}(\Ga_i-\Ga_j)$ is a square element in $\F_q$. The second method is to find a sufficient condition under which the homogenous equation system $A\bx^T=\bo$  with $A$ over $\F_q$ has a nonzero solution over $\F_r$, where $q=r^2$ (see \eqref{eq:2.3}).

\subsection*{Organization of the paper}
In Section 2, we first study generalized Reed-Solomon codes and their duals, and analyze solutions of   a system of homogenous equations.  In Section 3, we show that these conditions are satisfied in some cases and consequently we obtain several classes of MDS self-dual codes.

\section{Preliminaries}

Let $\F_q$ be the finite field of $q$ elements and let $\{\Ga_1,\Ga_2,\dots,\Ga_n\}$  be $n$ distinct elements of $\F_q$. Choose $n$ nonzero elements $v_1,v_2,\dots,v_n$ of $\F_q$ ($v_i$ may not be distinct). 
 Put $\bv=(v_1,v_2,\dots,v_n)$ and $\ba=(\Ga_1,\Ga_2,\dots,\Ga_n)$.
 Then the generalized Reed-Solomon code associated with $\ba$ and $\bv$ is defined below.
\begin{equation}\label{eq:2.1}\GRS_k(\ba,\bv):=\{(v_1f(\Ga_1),v_2f(\Ga_2),\dots,v_nf(\Ga_n)):\;f(x)\in\F_q[x],\ \deg(f(x))\le k-1\}.\end{equation}
It is well known that  the code $\GRS_k(\ba,\bv)$ is a $q$-ary $[n,k,n-k+1]$-MDS code \cite[Theorem 9.1.4]{LX} and the corresponding dual code is also a GRS code.

Furthermore we consider the extended code of the generalized Reed-Solomon code $\GRS_k(\ba,\bv)$ given by
\begin{equation}\label{eq:2.2}\GRS_k(\ba,\bv,\infty):=\{(v_1f(\Ga_1),v_2f(\Ga_2),\dots,v_nf(\Ga_n),f_{k-1}):\;f(x)\in\F_q[x],\ \deg(f(x))\le k-1\},\end{equation}
where $f_{k-1}$ stands for the coefficient of $x^{k-1}$ in $f(x)$. The following result can be easily derived from \cite{LX}.
\begin{lem}\label{lem:2.1} The code $\GRS_k(\ba,\bv,\infty)$ defined in {\rm \eqref{eq:2.2}} is a $q$-ary $[n+1,k,n+2-k]$-MDS code.
\end{lem}

For any distinct elements $\Ga_1,\dots,\Ga_n$ of $\F_{q}$, put $\ba= (\Ga_1,\dots,\Ga_n)$  and denote by $A_{\ba}$ the matrix
\begin{equation}\label{eq:2.4}\left(\begin{array}{cccc}
1 &1&\dots&1\\
\a_1&\a_2&\dots&\a_n\\
\a_1^2&\a_2^2&\cdots&\a_n^2\\
\vdots&\vdots&\ddots&\vdots\\
\a_1^{n-2}&\a_2^{n-2}&\cdots&\a_n^{n-2}\end{array} \right).\end{equation}

\begin{lem}\label{lem:2.1a} The solution space of the equation system $A_{\ba}\bx^T=\bo$ has dimension $1$ and  $\{\bu=(u_1,\dots,u_n)\}$ is a basis of this solution space, where  $u_i=\prod_{1\le j\le n,  j\neq i}(\Ga_i-\Ga_j)^{-1}$. Furthermore, for any two polynomials $f(x), g(x)\in\F_q[x]$ with $\deg(f)\le k-1$ and $\deg(g)\le n-k-1$, one has
$\sum_{i=1}^nf(\Ga_i)(u_ig(\Ga_i))=0$.
\end{lem}
\begin{proof} It is easy to see that the rank of $A_{\ba}$ is $n-1$. Thus, the solution space has dimension $1$. Furthermore, it is straightforward to verify that $\bu$ is a nonzero solution.

Since $\bu$ is a solution of $A_a\bx=\bo$, it is easy to see  that the Euclidean inner product of $(\Ga_1^i,\dots,\Ga_n^i)$ and $(u_1\Ga_1^j,\dots,u_n\Ga_n^j)$ is zero for all $0\le i\le k-1$ and $0\le j\le n-k-1$. This implies that  $\sum_{i=1}^nf(\Ga_i)(u_ig(\Ga_i))=0$  for any two polynomials $f(x), g(x)\in\F_q[x]$ with $\deg(f)\le k-1$ and $\deg(g)\le n-k-1$.
\end{proof}

\begin{lem}\label{lem:2.2} Let $\bi$ be all-one word of length $n$. Then  one has the following results.
\begin{itemize}
\item[(i)] The dual code of $\GRS_k(\ba,\bi)$ is $\GRS_{n-k}(\ba,\bu)$, where $\bu=(u_1,u_2,\dots,u_n)$ with $u_i=\prod_{1\le j\le n,  j\neq i}(\Ga_i-\Ga_j)^{-1}$.
\item[(ii)] If $1\le k\le q-1$, then the dual code of $\GRS_k(\ba,\bi,\infty)$ is $\GRS_{q-k+1}(\ba,\bi,\infty)$.
\end{itemize}
\end{lem}
\begin{proof} By the second statement of Lemma \ref{lem:2.1a}, we know that $\GRS_{n-k}(\ba,\bu)$ is orthogonal to $\GRS_k(\ba,\bi)$. Thus, part (i)  follows from the fact that $\dim(\GRS_k(\ba,\bi))+\dim(\GRS_{n-k}(\ba,\bu))=k+n-k=n$.

For part (ii), we denote
\[\bc_i =\left\{\begin{array}{ll}
(\Ga_1^i,\Ga_2^i,\dots,\Ga_q^i, 0) &\mbox{if $i\not=k-1$},\\
(\Ga_1^i,\Ga_2^i,\dots,\Ga_q^i, 1) &\mbox{if $i=k-1$};
\end{array}\right.\quad
\bbc_i =\left\{\begin{array}{ll}
(\Ga_1^i,\Ga_2^i,\dots,\Ga_q^i, 0) &\mbox{if $i\not=q-k$},\\
(\Ga_1^i,\Ga_2^i,\dots,\Ga_q^i, 1) &\mbox{if $i=q-k$}.
\end{array}\right.\]
Consider the dot product of $\bc_{\ell}$ and $\bbc_m$ with $0\le \ell\le k-1$ and $0\le m\le q-k$. If $\ell=m=0$,
 then both $\bc_{\ell}$ and $\bc_m$ are $(\bi, 0)$, where $\bi$ is the all-one word of length $q$. Thus,
  the dot product $\langle\bc_{\ell},\bbc_m\rangle$ is $0$. If $\ell=k-1$ and $m=q-k$, then
  $\bc_{\ell}=(\Ga_1^{\ell},\Ga_2^{\ell},\dots,\Ga_q^{\ell},1)$ and $\bbc_m=(\Ga_1^{m},\Ga_2^{m},\dots,\Ga_q^{m},1)$. Thus,
   $\langle\bc_{\ell},\bbc_m\rangle=1+\sum_{i=1}^q\Ga_i^{q-1}=0$. Now assume that $0<\ell+m<q-1$.
   Without loss of generality, let $\ell>0$. Then $\bc_{\ell}=(\Ga_1^{\ell},\Ga_2^{\ell},\dots,\Ga_q^{\ell},0)$.
   Thus,  $\langle\bc_{\ell},\bbc_m\rangle=\sum_{i=1}^q\Ga_i^{\ell+m}=0$ since $1\le \ell+m\le q-2$. This completes the proof of Part (ii).
\end{proof}
The  following corollary follows immediately from Lemma \ref{lem:2.2}.
\begin{cor}\label{cor:2.3} Let $n$ be an even number.
\begin{itemize}
\item[(i)] Let $\lambda\in\F_q^*$. If $w_i=\lambda\prod_{1\le j\le n,  j\neq i}(\Ga_i-\Ga_j)^{-1}$ is equal to $v_i^2$ for some $v_i\in\F_q$ for all $i=1,2,\dots,n$, then the code
 $\GRS_{n/2}(\ba,\bv)$ is MDS self-dual.
\item[(ii)] If $q$ is odd, then the  code  $\GRS_{(q+1)/2}(\ba,\bi,\infty)$ is self-dual.
\end{itemize}
\end{cor}
\begin{proof} To prove  Part (i), let $f(x), g(x)\in\F_q[x]$ with $\deg(f)\le \frac{n}2-1$ and $\deg(g)\le \frac{n}2-1$. By the second statement of Lemma \ref{lem:2.1a}, we have
$\sum_{i=1}^nf(\Ga_i)(u_ig(\Ga_i))=0$, where $u_i=\prod_{1\le j\le n,  j\neq i}(\Ga_i-\Ga_j)^{-1}$ for $i=1,2,\dots,n$. Hence,
\[
0=\lambda\sum_{i=1}^nf(\Ga_i)(u_ig(\Ga_i))=\sum_{i=1}^nf(\Ga_i)(\lambda u_ig(\Ga_i))=\sum_{i=1}^n(v_if(\Ga_i))(v_ig(\Ga_i)).
\]
This implies that $\GRS_{n/2}^{\perp}(\ba,\bv)=\GRS_{n/2}(\ba,\bv)$.

Part (ii) is a direct result of Lemma \ref{lem:2.2}(ii).
\end{proof}

For the rest of this section, we provide another sufficient condition under which a GRS code is self-dual. For this purpose,
we assume that $q=r^2$.

Let us consider a system of equations over $\F_{r^2}$ given by
\begin{equation}\label{eq:2.3}
A\bx^T=\bo,
\end{equation}
where $A$ is an $(n-1)\times n$ matrix of rank $n-1$ over $\F_{r^2}$. One knows that (\ref{eq:2.3}) must have at least one nonzero solution over $\F_{r^2}$. However, for our application, we are curious about the question  whether (\ref{eq:2.3}) has a nonzero solution over $\F_r$. In this section, we give some sufficient and necessary conditions under which (\ref{eq:2.3}) has a nonzero solution over $\F_r$.

\begin{lem}\label{lem:2.4} The equation (\ref{eq:2.3}) has a nonzero solution in $\F^n_r$ if and only if $\bc^r$ is a solution of (\ref{eq:2.3}) whenever $\bc$ is a solution of (\ref{eq:2.3}), where $\bc^r$ is obtained from $\bc$ by raising every coordinate to its $r$th power.
\end{lem}
\begin{proof} If (\ref{eq:2.3}) has a nonzero solution $\bb$ in $\F^n_r$, then the solution space of (\ref{eq:2.3}) is $\F_{r^2}\cdot\bb=\{\Ga \bb:\; \Ga\in\F_{r^2}\}$ since the solution space has dimension $1$ over $\F_{r^2}$. Thus, for every solution $\lambda\bb$, we have $(\lambda\bb)^r=\lambda^r\bb\in \F_{r^2}\cdot\bb$.

Conversely, assume that $\bc^r$ is a  solution of (\ref{eq:2.3}) for a nonzero solution $\bc$  of (\ref{eq:2.3}). Choose a basis $\{1,\a\}$ of $\F_{r^2}$ over $\F_r$. Consider the two elements $\bw_1:=\bc+\bc^r$ and $\bw_2:=\a\bc+\a^r\bc^r$. It is clear that both $\bw_1$ and $\bw_2$ are solutions of (\ref{eq:2.3}) in $\F^n_r$. On the other hand, we have
\[\left(\begin{array}{c}
\bc \\
\bc^r\end{array} \right)=\left(\begin{array}{cc}
1 & 1 \\
\a & \a^r\end{array} \right)^{-1}\left(\begin{array}{c}
\bw_1 \\
\bw_2\end{array} \right).\]
This implies that one of $\bw_1$ and $\bw_2$ must be nonzero, otherwise $\bc$ is equal to zero. This completes the proof.
\end{proof}

The condition given in Lemma \ref{lem:2.4} can be converted to a condition on the coefficient matrix of the equation (\ref{eq:2.3}) as shown below.

\begin{lem}\label{lem:2.5} Let $A$ be the coefficient matrix of the equation (\ref{eq:2.3}). Then  the equation (\ref{eq:2.3}) has a nonzero solution in $\F^n_r$ if and only if $A^{(r)}$ and $A$ are row equivalent, where $A^{(r)}$ is obtained from $A$ by raising every entry to its $r$th power.
\end{lem}
\begin{proof} It is easy to see that $\bc^r$ is a solution of $A^{(r)}\bx^T=\bo$ whenever $\bc$ is a solution of (\ref{eq:2.3}) and vice versa. By Lemma \ref{lem:2.1}, this implies that  the equation (\ref{eq:2.3}) has a nonzero solution in $\F^n_r$ if and only if the equation $A^{(r)}\bx^T=\bo$ and the equation (\ref{eq:2.3}) have the same solution space, i.e., $A^{(r)}$ and $A$ are row equivalent.
\end{proof}

\begin{exm}\label{exm:2.6}{\rm Let $m$ be a divisor of $r^2-1$ and let $n=m+1$. Let $\Ga_2,\dots,\Ga_n$ be all the $m$th roots of unity.
 We claim that the system $A_{\ba}\bx=\bo$ has a nonzero solution in $\F_r^n$, where $\ba=(\Ga_1=0,\Ga_2,\dots,\Ga_n)$. To prove this, it is sufficient
 to show that the rows of $A_{\ba}^{(r)}$ are a permutation of the rows of $A_{\ba}$.
The first row of the two matrices are identical. Hence, it is sufficient to show that the last $n-2=m-1$ rows of $A_{\ba}^{(r)}$ are  a permutation of those of $A_{\ba}$. To see this, we notice that the powers in the  last $m-1$ rows of $A_{\ba}^{(r)}$ consist of
 $\{1\cdot r,2\cdot r,\dots,(m-1)\cdot r\}$, while  the powers in the last $m-1$ rows of $A_{\ba}$ consist of $\{1,2\dots,m-1\}$.
 Thus, the desired result follows from
 the fact that the set $\{1\cdot r\pmod{m},2\cdot r\pmod{m},\dots,(m-1)\cdot r\pmod{m}\}$ and  the set $\{1,2\dots,m-1\}$ are identical.
}
\end{exm}

\begin{exm}\label{exm:2.7}{\rm  Let $\Ga_1,\dots,\Ga_n$ be all the $n$ distinct elements of $\F_r$. Then $A_{\ba}$ is a matrix over $\F_r$ and  it is clear that system $A_{\ba}\bx=\bo$ has a nonzero solution in $\F_r^n$. On the other hand, if we apply  Lemma \ref{lem:2.5}, we can also see that $A_{\ba}\bx=\bo$ has a nonzero solution in $\F_r^n$ since
$A_{\ba}$ and $A_{\ba}^{(r)}$ are
equal and hence row equivalent.
}
\end{exm}

\begin{lem}\label{lem:3.3} Let $n$ be an even number and let $q=r^2$.  Let $\Ga_1,\Ga_2,\dots,\Ga_n$ be $n$ distinct elements of $\F_{q}$.
 If $A_{\ba}\bx=\bo$ has a nonzero solution $\bw=(w_1,\dots,w_n)$ in $\F_r$, then the code $\GRS_{n/2}(\ba,\bv)$ is an MDS self-dual code over $\F_{q}$,
 where $w_i=v_i^2$ for all $1\le i\le n$.
\end{lem}
\begin{proof} Since $w_i$ belongs to $\F_r$, there exists an element $v_i\in\F_{q}$ such that $w_i=v_i^2$. As the dimension of the solution space of $A_{\ba}\bx=\bo$ is $1$,
by Lemma \ref{lem:2.1a}, we must have $w_i=\lambda\prod_{1\le j\le n,  j\neq i}(\Ga_i-\Ga_j)^{-1} \neq0$ for some $\lambda\in\F_{r^2}^*$. The desired result follows from Corollary \ref{cor:2.3}(i).
\end{proof}

\section{MDS self-dual codes}

\subsection{MDS self-dual codes over $\F_q$ for sufficiently large $q$}
Let us start with a lemma.
\begin{lem}\label{lem:3.1} For any given $n$,  if $q\ge 4^n\times n^2$, then there exists a subset $S=\{\Ga_1,\Ga_2,\dots,\dots,\Ga_n\}$ of $\F_q$  such that $\Ga_j-\Ga_i$
are nonzero square elements for all $1\le i<j\le n$.
\end{lem}
\begin{proof} If $q$ is even, it is clearly true as every element of $\F_q$ is a square.

Now assume that $q$ is odd. We prove it by induction on $n$. For $n=2$, we can let $S=\{0,1\}$. Suppose that there exists a subset $T=\{\Ga_1,\Ga_2,\dots,\dots,\Ga_{n-1}\}$ of $\F_q$ of size $n-1$ such that $\Ga_j-\Ga_i$
are nonzero square elements for all $1\le i<j\le n-1$.

Let $\Ga$ be a primitive element of $\F_q$ and let $\chi$ be the multiplicative quadratic character defined by $\chi(\Ga^i)=\Ga^{i(q-1)/2}$ and $\chi(0)=0$.
It is clear that $i$ is even if and only if $\chi(\Ga^i)=1$. Let $N$ denote the number of elements $\Gb$ of $\F_q$ such that $\chi(\Gb-\Ga_i)=1$ for all $i=1,2,\dots,n-1$. Then
by \cite[Exercise 5.64]{LN}, one has
\begin{equation}\label{eq:3.1}
\left|N-\frac{q}{2^{n-1}}\right|\le \left(\frac{n-3}{2}+\frac{1}{2^{n-1}}\right)\sqrt{q}+\frac{n-1}2.
\end{equation}
 Thus, by \eqref{eq:3.1} and our condition on $n$ and $q$, we have
\[N\ge \frac{q}{2^{n-1}}- \left(\frac{n-3}{2}+\frac{1}{2^{n-1}}\right)\sqrt{q}-\frac{n-1}2>0.\]
This implies that there exists an element $\Ga_n$ such that $\Ga_n-\Ga_i$ are nonzero square elements of $\F_q$ for all $i=1,2,\dots,n-1$. The proof is completed.
\end{proof}

\begin{thm}\label{thm:3.2} Let $n$ be an even integer. If $n$ and $q$ satisfy one of the following three conditions, then there exists a $q$-ary $[n,n/2,n/2+1]$-MDS self-dual code.
\begin{itemize}
\item[{\rm (i)}] $q$ is even and $n\le q$;
\item[{\rm (ii)}] $q\equiv 1\bmod{4}$, and $q\ge 4^n\times n^2$;
\item[{\rm (iii)}] $q$ is odd, $n=q+1$.
\end{itemize}
\end{thm}
\begin{proof} If $q$ is even, then every element of $\F_q$ is a square. Thus, Case (i) follows from Corollary \ref{cor:2.3}(i).

 By Lemma \ref{lem:3.1}, there
exists a subset $S=\{\Ga_1,\Ga_2,\dots,\Ga_n\}$ such that $\Ga_j-\Ga_i$ are square elements for all $1\le i<j\le n$. As
$q\equiv 1\bmod{4}$,  $-1$ is a square since $-1=\Ga^{(q-1)/2}$, where $\Ga$ is a primitive element of $\F_q$. Thus, $\Gb-\gamma$ is a nonzero square for
any two distinct elements $\Gb,\gamma\in S$. Therefore, $\prod_{1\le j\le n,  j\neq i}(\Ga_i-\Ga_j)^{-1}$ are nonzero square elements of $\F_q$ for all $i=1,2,\dots,n$.
Case (ii) follows from Corollary \ref{cor:2.3}(i) as well.

Case (iii) is the result of Corollary \ref{cor:2.3}(ii).
\end{proof}

\begin{rmk}\begin{itemize}
             \item [(i)]{\rm The results of Parts (i) and (iii) of Theorem \ref{thm:3.2} were given in \cite{GG08}. Here  a different proof is given.
}
             \item [(ii)] The result of Part (ii) of Theorem \ref{thm:3.2} implies that MDS self-dual code with length $n$ always exists when alphabet size $q$ is exponential in $n$.
           \end{itemize}
\end{rmk}
\subsection{MDS self-dual codes over $\F_q$ with $q=r^2$}

\begin{thm}\label{thm:3.4} Let $n$ be an even integer. If $n$ and $q=r^2$ satisfy one of the following two conditions, then there exists a $q$-ary $[n,n/2,n/2+1]$-MDS self-dual code.
\begin{itemize}
\item[{\rm (i)}] $n\le r$;
\item[{\rm (ii)}] $q$ is odd and $n-1$ is a divisor of $q-1$.
\end{itemize}
\end{thm}
\begin{proof} In case (i), we can choose $n$ distinct elements $\Ga_1,\Ga_2,\dots,\Ga_n$ of $\F_r$. Then the system $A_{\ba}\bx=\bo$ has a nonzero solution
in $\F_r^n$, By Lemma \ref{lem:3.3}, there exists a $q$-ary $[n,n/2,n/2+1]$-MDS self-dual code.

As $n-1$ is a divisor of $q-1$, by Example \ref{exm:2.6}, we can find $n$ distinct elements  $\Ga_1,\Ga_2,\dots,\Ga_n$ of $\F_{q}$ such that
 the system $A_{\ba}\bx=\bo$ has a nonzero solution in $\F_r^n$. Thus, Case (ii) follows from Lemma \ref{lem:3.3}.
This completes the proof.
\end{proof}

\begin{thm}\label{thm:3.5} Let $q=r^2$ and  $r\equiv 3\bmod{4}$, then there exists  a $q$-ary $[2tr,tr,tr+1]$-MDS self-dual code for any $1\le t\le (r-1)/2$.
\end{thm}
\begin{proof} Label elements of $\F_r$ by $\{a_1,a_2,\dots,a_{r}\}$. Assume that $\gamma$ is a primitive element of $\F_{q}$ and let $\Gb=\gamma^{(r+1)/2}$. Put $n=2tr$ and $\Ga_{\ell r+k}=a_{\ell}\Gb+a_k$ for all $1\le \ell\le 2t$ and $1\le k\le r$. By Corollary \ref{cor:2.3}(i), it is sufficient to show that $\prod_{1\le j\le n,  j\neq i}(\Ga_i-\Ga_j)^{-1}$ is a square of $\F_{q}$ for all $1\le i\le n$.

Write $i=\ell_0 r+k_0$ for some $1\le \ell_0\le 2t$ and $1\le k_0\le r$. Then
\begin{equation}\label{eq:3.2}v_{\ell_0}:=\prod_{\ell_0r+1\le j\le \ell_0r+r,  j\neq \ell_0r+k_0}(\Ga_{\ell_0 r+k_0}-\Ga_j)=\prod_{1\le j\le r,  j\neq k_0}(a_{k_0}-a_j)\in\F_r.\end{equation}
Thus, $v_{\ell_0}$ is a square in $\F_{q}$ since it is an element of $\F_r$. Furthermore, for $\ell\neq\ell_0$, we have
 \begin{equation}\label{eq:3.3}v_{\ell}:=\prod_{\ell r+1\le j\le \ell r+r}(\Ga_{\ell_0 r+k_0}-\Ga_j)=\prod_{1\le j\le r}((a_{\ell_0}-a_{\ell})\Gb+a_{k_0}-a_j)=((a_{\ell_0}-a_{\ell})\Gb)^r-(a_{\ell_0}-a_{\ell})\Gb=(a_{\ell_0}-a_{\ell})\Gb(\Gb^{r-1}-1).\end{equation}
This implies that $v_{\ell}$ is a square in $\F_{q}$ as well since $a_{\ell_0}-a_{\ell}$ and  $\Gb^{r-1}-1=-2$ are elements of $\F_r$ and $\Gb=\gamma^{(r+1)/2}$ is a square. Our result follows from the fact that $\prod_{1\le j\le n,  j\neq i}(\Ga_i-\Ga_j)^{-1}=\prod_{\ell=1}^{2t} v_{\ell}^{-1}$.
\end{proof}

\end{document}